\documentclass[11pt]{article}
\usepackage{fullpage}
\usepackage{amsfonts}
\usepackage{algorithmic}
\usepackage{algorithm}
\usepackage{amsmath,amsthm,amssymb}
\usepackage{latexsym}
\usepackage{epic}
\usepackage{epsfig}
\usepackage{amscd}
\usepackage{url}
\usepackage{verbatim}
\usepackage{tweaklist}

\newtheorem{theorem}{Theorem}[section]

\newtheorem{lemma}[theorem]{Lemma}

\newtheorem{claim}[theorem]{Claim}
\newtheorem{definition}[theorem]{Definition}

\newtheorem{observation}[theorem]{Observation}

\newenvironment{prevproof}[2]{\noindent {\em {Proof of
{#1}~\ref{#2}:}}}{\hfill $\square$\vskip \belowdisplayskip}

\renewcommand{\itemhook}{\setlength{\topsep}{4pt}%
  \setlength{\itemsep}{0pt} }
\renewcommand{\enumhook}{\setlength{\topsep}{4pt}%
  \setlength{\itemsep}{0pt} }

\newcommand{\set}[1]{\{ #1 \}}
\newcommand{\union}{\cup}
\newcommand{\intersect}{\cap}
\newcommand{\sm}{\setminus}

\newcommand{\ceiling}[1]{\lceil {#1} \rceil}

\DeclareMathOperator{\poly}{poly}

\newcommand{\RR}{\mathbb{R}}
\newcommand{\RRp}{\RR^+}

\newcommand{\NN}{\mathbb{N}}

\def\X{\mathcal{X}}
\def\R{\mathcal{R}}
\def\A{\mathcal{A}}

\def\C{\mathcal{C}}
\def\D{\mathcal{D}}
\def\V{\mathcal{V}}
\def\eps{\epsilon}
\def\sse{\subseteq}

\renewcommand{\comment}[1]{}

\newcommand{\hardness}{\compclass{NP}\not\sse
  \compclass{P/Poly}} 
\newcommand{\unhardness}{\compclass{NP}\sse \compclass{P/Poly}} 

\title{Amplified Hardness of Approximation for VCG-Based Mechanisms}
\author{Shaddin Dughmi \\
Department of Computer Science \\
Stanford University \\
{\tt shaddin@cs.stanford.edu} \and
Hu Fu\\
Department of Computer Science \\
Cornell University \\
{\tt hufu@cs.cornell.edu} \and
Robert Kleinberg \\
Department of Computer Science \\
Cornell University \\
{\tt rdk@cs.cornell.edu}}

\begin{document}

\maketitle \setcounter{page}{0} \thispagestyle{empty}

\begin{abstract}
  If a two-player social welfare maximization problem does not admit a
  PTAS, we prove that any maximal-in-range truthful mechanism that
  runs in polynomial time cannot achieve an approximation factor
  better than $1/2$.  Moreover, for the $k$-player version of the same
  problem, the hardness of approximation improves to $1/k$ under the
  same two-player hardness assumption.  (We note that $1/k$ is
  achievable by a trivial deterministic maximal-in-range mechanism.)
  This hardness result encompasses not only deterministic
  maximal-in-range mechanisms, but also all universally-truthful
  randomized maximal in range algorithms, as well as a class of
  strictly more powerful truthful-in-expectation randomized mechanisms
  recently introduced by Dobzinski and Dughmi.  Our result applies to
  any class of valuation functions that satisfies some minimal closure
  properties. These properties are satisfied by the valuation
  functions in all well-studied APX-hard social welfare maximization
  problems, such as coverage, submodular, and subadditive valuations.

  We also prove a stronger result for universally-truthful maximal-in-range
  mechanisms.  Namely, even for the class of budgeted additive
  valuations, which admits an FPTAS, no such mechanism can achieve an
  approximation factor better than $1/k$ in polynomial time.
\end{abstract}

\newpage

\newcommand{\valclass}{{\mathcal{C}}}
\newcommand{\compclass}[1]{{\mathsf{#1}}}

\section{Introduction} \label{sec:intro}

Do computational problems become harder when the inputs are supplied 
by selfish agents and the algorithm is required to operate in a way that
incentivizes truth-telling?  This question has been central to algorithmic
mechanism design since the field's inception~\cite{NR01}.  
The most famous positive
result in the area is also one of the simplest: any efficient 
social-welfare-maximization algorithm can be transformed into a 
computationally efficient truthful mechanism using the celebrated
VCG payment scheme~\cite{Cla71,Gro73,Vic61}.  
It is also well known that this
result does not extend to approximation algorithms: in order 
for an algorithm to be truthfully implemented by the VCG
payment scheme, the algorithm must satisfy a property known
as \emph{maximal-in-range} (MIR)~\cite{NR07}, which is
unfortunately violated by most approximation algorithms.
However, the technique of combining a maximal-in-range algorithm with
the VCG payment scheme remains the only known general-purpose
technique for designing truthful mechanisms for multi-parameter
domains\footnote{A multi-parameter domain is one in which an agent's 
private information consists of more than just a single real-valued
parameter.}, and consequently a great deal of research has been
devoted to searching for computationally efficient approximation
algorithms that are maximal-in-range, e.g.~\cite{DN07b,DNS06,LS05},
or proving hardness-of-approximation theorems for this class
of algorithms, e.g.~\cite{DN07a,MPSS09,PSS08}.  

Combinatorial auctions are the most well-studied, and arguably the
most important, class of mechanism design problems, and they furnish
striking insights into the capabilities and limitations of
maximal-in-range mechanisms.  One can bound the approximation ratio of
a combinatorial auction mechanism in terms of many parameters,
including the number of players (henceforth denoted by $k$) or the
number of items (henceforth, $m$).  A particularly bleak picture
emerges when one bounds the approximation ratio in terms of the number
of players.  Any combinatorial auction with $k$ players has a trivial
mechanism that simply packages all the items as a single bundle and
awards this bundle to the bidder who values it most highly.  This
mechanism is computationally efficient, maximal in range, but only
achieves approximation ratio $1/k$.  Despite years of research on
truthful combinatorial auctions, to our knowledge this dependence on
$k$ can be improved in only one combinatorial auction domain: the
domain of multi-unit auctions.  Here, the underlying social welfare
maximization problem admits a deterministic FPTAS that is not
maximal-in-range.  Dobzinski and Nisan~\cite{DN07b} discovered a
deterministic $2$-approximation that is MIR, and Dobzinski and
Dughmi~\cite{DD09} discovered an FPTAS that outputs a randomized
allocation satisfying a property called \emph{maximal in
  distributional range} (MIDR).  In fact, the Dobzinski-Dughmi
mechanism satisfies a stronger property that we call \emph{maximal in
  weighted range} (MIWR), meaning that the only use of randomization
is to cancel the allocation (i.e., allocate no items) with some
probability.

\subsection{Our contributions}
\label{sec:contributions}

In this paper we show that there is an inherent reason why 
truthful mechanisms cannot break the ``$1/k$ barrier'' for
combinatorial auctions: \emph{any approximation hardness
at all in the underlying social welfare maximization problem
is amplified to $(1/k + \eps)$-hardness when one
restricts the algorithm to be maximal in range.}  In fact,
our result extends to two classes of randomized 
MIDR mechanisms: those that choose an allocation deterministically
and then toss coins to decide whether to cancel the allocation
(MIWR mechanisms) and those that toss coins to choose a 
deterministic MIR mechanism and then run it
(randomized MIR mechanisms), as well as the combination of
the two.
Our result applies to any class of valuations that is 
\emph{regular}, meaning that it satisfies some natural
closure properties to be specified in Section~\ref{sec:preliminaries}.
These properties are
satisfied by the valuation functions in all well-studied APX-hard
social welfare maximization problems, such as coverage,
submodular, XOS, and subadditive valuations.

\begin{theorem} 
Let $\valclass$ be any regular class of valuations such that two-player
social welfare maximization with valuations in $\valclass$
does not admit a PTAS.  Then for all $\eps>0$,
there is no polynomial-time 
randomized MIWR mechanism that achieves an
approximation ratio greater than $1/k + \eps$ unless
$\compclass{NP} \sse \compclass{P}/\poly$.
\end{theorem}

%

Even when social welfare maximization over $\valclass$
admits an FPTAS, it may still be possible to prove
that maximal-in-range mechanisms cannot improve on 
the trivial $1/k$ approximation factor.  In fact, we are
able to show this for an important class of valuations
that admits an FPTAS: the class of \emph{budgeted additive
valuations}, in which each player $i$ has a budget $B_i$,
and her value for a bundle is equal to the sum of her
values for the individual items, or to $B_i$, whichever
is smaller.  However, unlike Theorem~\ref{thm:main}, this
hardness result is limited to universally-truthful randomized maximal-in-range
mechanisms.

\begin{theorem} \label{thm:bav}
For all $\eps>0$,  no polynomial-time 
randomized MIR mechanism for
combinatorial auctions with budgeted additive
valuations can achieve an
approximation ratio greater than $1/k + \eps$, unless
$\compclass{NP} \sse \compclass{P}/\poly$.
\end{theorem}

\subsection{Our techniques}
\label{sec:techniques}

A paradigm for proving hardness results of this sort was
introduced by Papadimitriou, Schapira, and Singer
in~\cite{PSS08}.  
To prove that a certain approximation ratio
cannot be achieved by maximal-in-range mechanisms, one
proves that the underlying social welfare maximization 
problem exhibits a particularly strong form of 
self-reducibility: any maximal-in-range algorithm 
for optimizing over a sufficiently large subset
of allocations can 
be transformed into an algorithm for optimizing over 
\emph{all allocations} of a smaller set of items.
The technical core of any such proof is a lemma
showing that any sufficiently large range of allocations
must ``shatter'' a fairly large subset $S$ of the items,
meaning that there is a set of players $P$ such 
that all allocations of $S$ to $P$ occur in the
range.  In~\cite{PSS08} the relevant shattering
lemma was the famous Sauer-Shelah Lemma.  But since
the Sauer-Shelah Lemma is a statement about collections
of a \emph{subsets} of a ground set $U$, and 
the range of a combinatorial auction is a 
collection of \emph{partial functions} from $U$
to the set of players, we require new shattering lemmas 
that apply to partial functions.

Extending the Sauer-Shelah Lemma to partial functions
is far from trivial: a lower bound on the cardinality
of the range does not suffice to prove that it shatters
a large set of items; for example, the set of all
allocations that give a subset of the items to
player $1$ and no items to any other player constitutes
an exponentially large range but does not shatter
any nonempty set of items.  Thus, one needs to 
carefully define what is meant by the hypothesis 
that the range is ``large'', and also (in the case
of more than two players) what is meant by the
conclusion that it ``shatters'' a large set of items.
In this paper, we provide two such lemmas.  In both
of them, $U$ and $V$ are finite sets with $|U|=m, \,
|V|=k,$ and $R$ is a set of functions from $U$ 
to $V \cup \set{\ast}.$

\begin{lemma} \label{lem:shatter-covering}
Suppose that for a random $f:U \rightarrow V$,
with probability at least $\gamma$ there
is a $g \in R$ such that $g(x)$ differs from $f(x)$
on at most $\left(1-\frac{q-1}{k}-\eps\right)m$ 
elements $x \in U$.  Then there is a subset $S \subseteq U$ of
cardinality at least $\delta m$ (where $\delta>0$
may depend on $\gamma,\eps,q,k$) and
a subset $T \subseteq V$ of cardinality $q$,
such that every function from $S$ to $T$ occurs
as the restriction of some $g \in R$.
\end{lemma}

\begin{lemma} \label{lem:shatter-eps-far}
Suppose that $\eps,\alpha,\ell$ are constants such that 
$|R| > e^{\alpha m}$, and suppose that for every $\ell$-tuple of functions 
$g_1,\ldots,g_\ell \in R$, for some $1 \leq i < j \leq \ell$
there are at least $\eps m$ elements 
$x \in U$ such that $g_i(x)$ and $g_j(x)$ are distinct
elements of $V$.  Then there is a subset $S \subseteq U$
of cardinality at least $\delta m$ (where $\delta>0$
may depend on $\eps,\alpha,\ell,k$) and a pair of
elements $a,b \in V$, such that every function from
$S$ to $\{a,b\}$ occurs as the restriction of 
some $g \in R$.
\end{lemma}

The first lemma, which underlies our proof of Theorem~\ref{thm:main}
and may be of independent interest,
substitutes an assumption that $R$ has small covering radius
in the Hamming metric in place of the usual assumption that
$R$ has large cardinality.  The second lemma, which  
underlies our proof of Theorem~\ref{thm:bav}, 
generalizes and closely parallels a related lemma 
from~\cite{MPSS09}. We prove both lemmas in Appendix \ref{app:shatter}.

To derive the lower bound for MIWR mechanisms, an
additional idea is needed: rather than reducing 
directly converting an $\alpha$-approximate MIR
algorithm into an exact optimization algorithm 
over a smaller set of items, we convert an
$(\alpha-\delta)$-approximate MIWR 
algorithm into an $(\alpha+\delta)$-approximate 
MIWR algorithm over a smaller set of items, 
and then we reach a contradiction
by taking $\alpha$ to be the supremum of the 
approximation ratios achievable by polynomial-time
MIWR mechanisms.  Translating this idea into a rigorous
proof requires a delicate induction over the number
of players.  Finally, to extend the result to randomized MIWR 
mechanisms, we show that any randomized MIWR mechanism 
can be transformed into a
MIWR mechanism with polynomial advice, incurring a 
negligible loss in the approximation factor.  The
proof of this step closely parallels Adleman's proof
that $\compclass{BPP} \sse \compclass{P}/\poly$.

\subsection{History of these results}
\label{sec:history}

Our work builds on the work of~\cite{MPSS09}, which obtained a weaker
version of Theorem~\ref{thm:bav}, also using the ``shattering''
technique.  Independently and concurrently with our discovery of
Theorem~\ref{thm:bav}, a different proof of a similar result (limited
to deterministic mechanisms, but obtaining optimal dependence on the
number of items as well as players) was discovered by Buchfuhrer and
Umans~\cite{BU09}.  Our Lemma~\ref{lem:shatter2}, which constitutes a
step in the proof of Theorem~\ref{thm:main} and was discovered after
we had read the proof of the Buchfuhrer-Umans result, uses a counting
argument similar to their proof of a seemingly unrelated shattering
lemma in~\cite{BU09}.



\section{Preliminaries}
\label{sec:preliminaries}

We assume the reader is familiar with standard terminology and
notation regarding truthful mechanisms and approximation algorithms.
Appendix~\ref{app:prelim} contains the relevant definitions.

\subsection{Combinatorial Auctions}
In a combinatorial auction there is a set $[m]=\set{1,2,\ldots,m}$ of
items, and a set $[k]=\set{1,2,\ldots,k}$ of players. Each player $i$
has a valuation function $v_i:2^{[m]}\rightarrow \RRp$ that is normalized
($v_i(\emptyset) = 0$) and monotone ($v_i(A) \leq v_i(B)$ whenever $A
\sse B$).  

An allocation of items $M$ to the players $N$ is a function $S: M \to
N \union \set{*}$. Notice that we do not require all items to be
allocated. If an allocation $S$ allocates all items -- i.e. $S$ maps $M$
into $N$ -- we say $S$ is a \emph{total allocation}. The allocation that
allocates no items is called the \emph{empty allocation}. For
convenience, we use $S(j)$ to denote the player receiving item $j$,
and we use $S_i$ to denote the items allocated to player $i$. We use
$\X(M,N)$ to denote the set of all alocations of $M$ to $N$.

In combinatorial auctions, the feasible solutions are the allocations
$\X([m],[k])$ of the items to the players. The \emph{social welfare}
of such an allocation $S$ is defined as $\sum_i v_i(S_i)$. When the
players have values $\set{v_i}_i$, we often use $v(S)$ as
shorthand for the welfare of $S$. The goal in combinatorial auctions
is to find an allocation that maximizes the social welfare.

\subsection{Valuation Classes}\label{subsec:vals}
The hardness of designing truthful combinatorial auction mechanisms
depends on the allowable player valuations. Recall that a
\emph{valuation} over $M$ is a function $v: 2^M \to \RRp$. We let
$\mathcal V$ denote the set of all valuations over all abstract finite sets
$M$. A \emph{valuation class} $\C$ is a subset of $\mathcal
V$. Examples of valuation classes include submodular valuations,
subadditive valuations, single-minded valuations, etc.  Our first
result applies to any valuation class that satisfies some natural properties. 

\begin{definition}\label{def:regular}
  We a say a valuation class $\C$ is \emph{regular} if the following hold
\begin{enumerate} 
\item Every valuation in $\C$ is monotone and normalized.
\item The canonical valuation on any singleton set is in $\C$. Namely,
  for any item $a$ the valuation $v:2^{\set{a}} \to \RRp$, defined as
  $v(\set{a}) = 1$ and $v(\emptyset)=0$, is in $\C$.
\item Closed under scaling: Let $v:2^M \to \RRp$ be in $\C$, and let $c\geq
  0$. The valuation $v': 2^M \to \RRp$, defined as $v'(A) = c \cdot v(A)$ for all $A
  \sse M$, is also in $\C$.
\item Closed under disjoint union: Let $M_1$ and $M_2$ be disjoint
  sets. Let the valuations $v_1: 2^{M_1} \to \RRp$ and $v_2:2^{M_2} \to
  \RRp$ be in $\C$. Their disjoint union $v_3 = v_1 \oplus v_2 : 2^{M_1
    \union M_2} \to \RRp$,
  defined as $v_3(A) = v_1(A \intersect M_1) + v_2( A \intersect M_2)$
  for all $A \sse M_1 \union M_2$, is in $\C$.
\item Closed under relabeling: Let $M_1,M_2$ be sets with a
  bijection $f : M_2 \rightarrow M_1.$  If $v_1 : 2^{M_1} \to \RRp$
  is in $\C$, then the valuation $v_2 : 2^{M_2} \to \RRp$ defined
  by $v_2(S) = v_1(f(S))$ is also in $\C$.
\end{enumerate}
\end{definition}
Note that all regular valuation classes support \emph{zero-extension}.
More formally, let $M \sse M'$, and let $v: 2^M \to \RRp$ be in
$\C$. The extension of $v$ to $M'$, defined as $v'(A) = v(A \intersect
M)$ for all $A \sse M'$, is also in $\C$. In the context of
combinatorial auctions, we use $\C_m$ to denote the subset of
valuation class $\C$ that applies to items $[m]$.

Most well-studied valuation classes for which the underlying
optimization problem is APX-hard are regular. This includes
submodular, subadditive, coverage, and weighted-sum-of-matroid-rank
valuations. However, two interesting counter-examples come to mind:
multi-unit (where items are indistinguishable), and single-minded
valuations. Nevertheless, the underlying optimization problem is not
APX hard for multi-unit valuations, and for single-minded valuations
the computational hardness of approximation is $1/k^{1-\epsilon}$ even
without the extra constraint of truthfulness
(see \cite{BN07}).

Our second hardness result pertains to
deterministic mechanisms for a very simple, non-regular class:
budgeted additive valuations. This is despite the fact that the
underlying $k$-player optimization problem admits an FPTAS
\cite{AndelmanM04}. Budgeted additive valuations are defined as follows.

\begin{definition}
  We say a valuation $v: 2^M \to \RRp$ is \emph{budgeted additive} if
  there exists a constant $B\geq 0$ (the budget) such that $v(A) =
  \min(B,\sum_{i \in A} v(\set{i}))$.
\end{definition}

\subsection{MIR, Randomized MIR, MIDR, and MIWR}

Maximal in range (MIR) algorithms were introduced in \cite{NR07} as a
paradigm for designing truthful approximation mechanisms for
computationally hard problems. An algorithm $\A$ is maximal-in-range
if it induces a maximal-in-range allocation rule when $k$ and $m$ are fixed.

\begin{definition}
  A $k$-bidder, $m$-item allocation rule $f$ is
  \emph{maximal-in-range} (MIR) if there exists a set of allocations
  $\R \sse \X([m],[k])$, such that $\forall v_1, \ldots, v_k \;
  f(v_1,...,v_k) \in \arg\max_{S \in \R} \Sigma_i v_i(S_i)$.
\end{definition}

A generalization of maximal-in-range that uses randomization sometimes
yields better algorithm. An algorithm $\A$ is \emph{randomized
  maximal-in-range} if it induces a maximal-in-range allocation rule
for every realization of its random coins. It is well known a
randomized MIR algorithm can be combined with the VCG payment scheme
to yield universally truthful mechanisms.

Dobzinski and Dughmi defined a generalization of randomized maximal-in-range
algorithms in \cite{DD09}, termed maximal-in-distributional-range
(MIDR). Here, each element of the range is a distribution over
allocations. The resulting mechanism outputs the distribution in the
range that maximizes the expected welfare, and charges VCG payments.

\begin{definition}
  $f$ is \emph{maximal-in-distributional-range} (MIDR) if there exists
  a set $\D$ of distributions over allocations such that for all
  $v_1, \ldots, v_k$, $f(v_1,...,v_k)$ is a distribution $D \in \D$
  that maximizes the expected welfare of a random sample from $D$.
\end{definition}

MIDR algorithms were used in \cite{DD09} to obtain a polynomial-time
truthful-in-expectation $\mbox{FPTAS}$ for multi-unit auctions,
despite a lower bound of 2 on polynomial-time maximal-in-range
algorithms. Moreover, they exhibited a variant of multi-unit auctions
for which an MIDR FPTAS exists, yet no deterministic (or even
universally truthful) polynomial time mechanism can attain an
approximation ratio better than 2. Notably, the MIDR algorithms
presented in \cite{DD09} are of the following special form.

\begin{definition}
  $f$ is \emph{maximal in weighted range} (MIWR) if $f$ is MIDR, and
  moreover each distribution $D$ in the range of $f$ is a
  \emph{weighted allocation}: There is a pure allocation $S \in \X([m],[k])$
  such that $D$ outputs $S$ with some probability, and the empty
  allocation otherwise.
\end{definition}

We denote a weighted allocation that outputs $S$ with probabiliby $w$
by the pair $(w,S)$. When there is room for confusion, we use
the term \emph{pure allocation} to refer to an unweighted allocation.

Our first result will apply to all polynomial time MIWR mechanisms,
and is the first such negative result. In fact, this result also
applies to any randomization over MIWR mechanisms, a class we term
\emph{randomized MIWR} mechanisms. Randomized MIWR mechanisms
include all universally-truthful randomized MIR mechanisms as a special
case. Our second result will apply to only randomized maximal-in-range
mechanisms, yet applies to a very restricted class of valuations,
namely budgeted-additive valuations.

\subsection{Some Complexity Theory}
\label{sec:complexity}

\newcommand{\circuit}{{\mathcal{C}}}

Broadly speaking, our proof involves constructing a reduction that
transforms every instance of a $k$-player mechanism design problem into 
an instance of one of $k$ other problems 
$\mathcal{P}_1,\ldots,\mathcal{P}_k$, each of which individually
is presumed to be computationally hard.  The reduction has the
property that input instances with a given number of items, $m$,
are all transformed into inputs of the same problem $\mathcal{P}_i$,
but instances with a different number of items may map to a 
different one of the $k$ problems.  This raises difficulties 
because the complexity of $\mathcal{P}_1,\ldots,\mathcal{P}_k$
may be ``wild'': for each of them, there may be some input
sizes (perhaps even infinitely many) that can be solved by
a polynomial-sized Boolean circuit.  In this section we 
develop the relevant complexity-theoretic tools to surmount
this obstacle. We relegate the proofs of these results to Appendix \ref{app:complexity}.

\begin{definition}
A set $S \subseteq \NN$ is said to be \emph{complexity-defying (CD)}
if there exists a family of polynomial-sized Boolean circuits
$\{\circuit_n\}_{n \in \NN}$ such that for all $n \in S$, 
the circuit $\circuit_n$ correctly decides {\sc 3sat} on all
instances of size $n$.

A set $T \subseteq \NN$ is said to be 
\emph{polynomially complexity-defying (PCD)} if there exists
a complexity-defying set $S$ and a polynomial function $p(n)$
such that $T \subseteq \bigcup_{n \in S} [n,p(n)].$
Here $[a,b]$ denotes the set of all natural numbers
$x$ such that $a \leq x \leq b$.  If a set $U \subseteq \NN$
is not PCD, we say it is \emph{non-PCD}.
\end{definition}

\begin{lemma}  A finite union of CD sets is CD, and a finite union
of PCD sets is PCD.
\label{lem:finite-union}
\end{lemma}

\begin{definition}
A decision problem or promise problem is said to have the 
\emph{padding property} if for all $n < m$ 
there is a reduction that transforms instances
of size $n$ to instances of size $m$, running
in time $\poly(m)$ and mapping ``yes'' instances
to ``yes'' instances and ``no'' instances to ``no''
instances.  Similarly, an optimization problem is 
said to have the padding property if for all
$n < m$ there is a reduction that transforms instances
of size $n$ to instances of size $m$, running
in time $\poly(m)$ and preserving the optimum value of
the objective function.
\end{definition}

\begin{lemma} \label{lem:pcd}
Suppose that $\mathcal{L}$ is a decision problem 
or promise problem that has the padding property
and is NP-hard under polynomial-time
many-one reductions.  Let $T$ be any subset of $\NN$.  
If there is a polynomial-sized circuit family that decides
$\mathcal{L}$ correctly whenever the input size belongs
to $T$, then $T$ is PCD.
\end{lemma}

\begin{lemma} \label{lem:not-pcd}
If $\NN$ is PCD, then $\compclass{NP} \subseteq \compclass{P}/\poly.$
\end{lemma}

\subsection{Technical Assumptions For Main Result}
\label{subsec:assumptions}
For our main result, a note is in order on the representation of
valuation. Our results hold in the computational model. Therefore, we
may assume that valuation functions are \emph{succint}, in that they
are given as part of the input, and can be evaluated in time
polynomial in the length of their description. Naturally, our main
result applies to non-succint valuations with oracle access, when the
resulting problem admits a suitable reduction from an APX hard
optimization problem.

Moreoever, due to the generality of our main result, we need to make
some technical assumptions. Namely, we restrict our attention to
Combinatorial Auctions over a well-behaved family of instances. This
restriction is without loss of generality for all well-studied classes
of valuations for which the problem is APX-hard, such as coverage,
submodular, etc. A family $I$ of inputs to Combinatorial auctions is
\emph{well-behaved} if there exists a polynomial $b(m)$ such that for
each input $(k,m, v_1,\ldots,v_k) \in I$, the function $v_i$ is
represented as a bit-string of length $O(b(m))$, and moreover always
evaluates to a rational number with $O(b(m))$ bits. While we believe
this assumption may be removed, we justify it on two grounds: First,
every well-studied variant of combinatorial auctions that is APX hard
is also APX hard on a well-behaved family of instances, so this
restriction is without loss for all such variants. Second, this assumption greatly
simplifies our proof, since it allows us to describe the size of an
instance by a single parameter, namely $m$.






\section{Amplified Hardness for APX-Hard Valuations}
\label{sec:midr_apx}

In this section, we prove the following main result.

\begin{theorem}\label{thm:main}
  Fix a regular valuation class $\C$ 
 for which $2$-player social welfare maximization is
  $APX$-hard. Fix a constant $k\geq 1$. For any constant $\epsilon >
  0$, no polynomial-time randomized MIWR algorithm for $k$-player
  combinatorial auctions achieves an expected approximation ratio of
  $1/k + \epsilon$, unless $\unhardness$.
\end{theorem}

It is worth noting that this impossibility result applies to all
universally-truthful randomized maximal-in-range algorithms. First, we
prove the analogous result for MIWR mechanisms that take polynomial
advice.

\begin{theorem}\label{thm:main_poly}
  Fix a regular valuation class $\C$ 
for which $2$-player social welfare maximization is
  $APX$-hard. Fix a constant $k\geq 1$. For any constant $\epsilon >
  0$, no non-uniform polynomial-time MIWR algorithm for $k$-player
  combinatorial auctions achieves an expected approximation ratio of
  $1/k + \epsilon$, unless $\unhardness$.
\end{theorem}

We then complete the proof by showing that any randomized MIWR
mechanism can be ``de-randomized'' to one that takes polynomial
advice.

A word is in order on the notion of \emph{non-uniform}
computation. For this, the reader should refer to Appendix
\ref{sec:nonuniform}.   Our hardness results in this section
follow from the commonly-held conjecture that non-uniform computation
cannot solve NP-complete problems, in other words $\hardness$.


Our proof strategy is as follows. In Section \ref{subsec:perfect} we
define a ``perfect valuation profile'' on $k$ players as a set of
valuations where exactly one player is interested in each item. We
then show that any range of allocations that gives a good
approximation on a randomly drawn perfect valuation profile must
``shatter'' a constant fraction of the items, meaning that the range
contains all allocations of that subset of the items to $q$ of the
players, where the value of $q$ depends on the quality of the
approximation.  (A better approximation implies a larger $q$.)

In Section \ref{subsec:main_poly}, we prove Theorem
\ref{thm:main_poly} by induction on the number of players $k$. Roughly
speaking, we show that for any MIWR mechanism $\A$ for $k$ players,
the allocations with weight much larger than $1/k + \epsilon$ are
useless. Namely, the inductive hypothesis implies that the allocations
with weight sufficiently larger than $1/k + \epsilon$ cannot yield a
good approximation to a randomly drawn perfect valuation; otherwise,
one could use the resulting shattered set of items to design a
strictly better MIWR mechanism for $k'$ players for some $k' <
k$. This allows us to conclude that all ``useful'' allocations have very
similar weight to one another; within $1- \eta$ for arbitrarily small
$\eta$ and a sufficiently large set of items. Since the mechanism
maximizes over a large set of allocations that are almost ``pure'', in
the sense that the weights are almost identical, this yields a PTAS,
contradicting the APX hardness of the problem.

Finally, we complete the proof of Theorem \ref{thm:main} in Section
\ref{subsec:mainresult}, using a de-randomization argument. This step is
similar to Adleman's proof that $\compclass{BPP} \sse
\compclass{P/Poly}$.

\subsection{Perfect Valuations} \label{subsec:perfect}
We define a perfect valuation profile as one where each item is
desired by exactly one player. Perfect valuation profiles will prove
useful in our proof, due to the fact that no ``small'' range can well-approximate
social welfare for a randomly-drawn perfect valuation profile.

\begin{definition}
  Let $N$ and $M$ be a set of players and items, respectively. Let
  $v_i: 2^M \to \RRp$ be the valuation of player $i\in N$. We say the
  valuation profile $\set{v_i}_{i\in N}$ is a \emph{perfect valuation
    profile} on  $N$ and $M$ if there exists a total
  allocation $S$ of $M$ to $N$ such that $v_i(j) = 1$ if $j \in S_i$,
  and $v_i(j) = 0$ otherwise. In this case, we say that
  $\set{v_i}_{i\in N}$ is the perfect valuation profile
  \emph{generated} by $S$.
\end{definition}

To use perfect valuations in our proof, they must be allowable
valuations. Indeed, it follows immediately from definition
\ref{def:regular} that any regular class of valuations contains all
perfect valuations.


The key property of perfect valuations is a reinterpretation of Lemma
\ref{lem:shatter-covering}, and can be summarized as follows.  If a
range $\R$ of allocations achieves a ``good'' approximation for many
perfect valuations, then $\R$ must include all allocations of a
constant fraction of the items to some set of $q$ players. Here, the
number of players $q$ depends on the quality of the approximation
guaranteed by $\R$, with a better approximation yielding a larger
$q$. The precise dependence of $q$ on the quality of the
approximation, as stated in Lemma \ref{lem:shatter-covering}, will
prove key in Section \ref{subsec:main_poly}. 


\subsection{Hardness for Non-Uniform MIWR Mechanisms} \label{subsec:main_poly}

In this section, we prove Theorem
\ref{thm:main_poly}, assuming $\hardness$. We fix the valuation class
$\C$ as in the statement of the theorem. Moreover, we fix $\eta > 0$
such that the 2-player social welfare maximization problem
is APX-hard to approximate within $1-\eta$. The
proof proceeds by induction on $k$. We need the following strong
inductive hypothesis:


\newcommand{\ih}[1]{{\bf IH($\bf #1$)}}
\newtheorem*{ind}{\ih{k}}
\begin{ind}
  For any constant $\alpha > 1/k$ and set $T \subseteq \NN$, 
  if a non-uniform polynomial-time MIWR
  algorithm for the $k$-player problem achieves an
  $\alpha$-approximation for $m$ items whenever $m \in T$,
  then $T$ is PCD.  
\end{ind}

In other words, the set of input lengths for which any particular such
algorithm may achieve an $\alpha$-approximation is PCD.
(See Section~\ref{sec:complexity} for  
the definition of PCD.)
It is clear that establishing \ih{k} for all $k\geq 1$ proves Theorem
\ref{thm:main_poly}, since $\NN$ is not PCD. 
The base case of $k=1$ is trivial. We now fix
$k$, and assume \ih{q} for all $q < k$.

\newcommand{\FF}{\mathbb{F}}

Assume for a contradiction that \ih{k} is violated for some
$\alpha$. Let $\alpha > 1/k$ be the supremum over all
values of $\alpha$ violating it. Note that \ih{k-1} implies that
$\alpha \in \left(\frac{1}{k}, \frac{1}{k-1} \right]$. To simplify the
exposition, we assume the supremum is attained, and fix the algorithm
$\A$ (and corresponding family of polynomial advice strings)
achieving an $\alpha$-approximation for all $m \in \FF$ where
$\FF$ is not PCD.  Our arguments can all be
easily modified to hold when the supremum is not attained, by
instantiating $\A$ to achieve $(\alpha-\zeta)$ instead, where
$\zeta>0$ is as small as needed for the forthcoming proof. The proof
then proceeds as follows.  Letting $\D^m$ denote the range of $\A$
when the number of items is $m$, we partition $\D^m$ into $k$ sets
$\D^m_q \, (2 \leq q \leq k+1)$ according to the weight of the
allocation.  
We also assign every $m \in \FF$ to one or more subsets
$T_q \, (2 \leq q \leq k+1)$; the definition of $T_q$
is quite technical, but roughly speaking $m \in T_q$
if the output of $\A$, when applied to a random perfect
valuation profile, has probability at least $1/k$ of being
in $\D^m_q$.  As we said, $\FF = \bigcup_{q=2}^{k+1} T_q.$
However, we will prove that $T_q$ is PCD
for all $q$, hence by Lemma~\ref{lem:finite-union} their
union $\FF$ is PCD.  This contradicts our earlier assumption
that $\FF$ is not PCD and completes the proof.

To prove that $T_q$ is PCD, we distinguish three cases depending on
the value of $q$.  If $2 \leq q \leq k-1,$ then we prove that $m \in
T_q$ implies that there is a non-uniform polynomial-time MIWR
mechanism for the $q$-player problem that achieves an approximation
ratio strictly better than $1/q$ when the number of items is
$\ceiling{\sigma m},$ for some constant $\sigma>0$.  By our induction
hypothesis, the set of all such $\ceiling{\sigma m}$ is a PCD set.  If
$q=k$, then we proceed similarly but working with the $k$-player
problem and proving an approximation ratio strictly better than
$\alpha$ when the number of items is $\ceiling{\sigma m}$; once again
this implies that the set of all such $\ceiling{\sigma m}$ is a PCD
set, by our hypothesis on $\alpha$.  Finally, if $q=k+1$ then we prove
that there is a non-uniform polynomial-time algorithm achieving
approximation ratio $1-\eta$ for the two-player social welfare
maximization problem, where $\eta$ was chosen so that the problem APX
hard to approximate within $1-\eta$.  By Lemma~\ref{lem:pcd}, this
implies $T_q$ is PCD.

\paragraph{Defining the partition of the range.}
Recall that an MIWR mechanism fixes a range of weighted allocations
for each $m$. Let $\D^m$ denote the range of $\A$ when the number of
items is $m$. Let $\R^m=\set{S \in \X([m],[k]) : (w,S) \in \D^m \mbox{
    for some $w$}} $ be the corresponding set of pure allocations. For
each allocation $S \in \R^m$, we use $w(S)$ to denote the weight of
$S$ in $D^m$. We assume without loss of generality that there is a
unique choice of $w(S)$, since allocations with greater weight are
always preferred. When $m\in \FF$, we may assume without loss of
generality that $w(S) \geq \alpha$ for every $S \in \R^m$, since $\A$
achieves an $\alpha$ approximation for $m$. We fix $\epsilon>0$ such
that $\alpha > 1/k+\epsilon,$ and $\xi>0$ such that $1/k +
\epsilon/2 = (1-\xi)^{-1} \cdot (1/k),$, and let $\delta = \epsilon
\eta / 5k$.  We partition $\R^m$ into \emph{weight classes} as
follows:
  \begin{itemize}
  \item $\R^m_q = \set{S \in \R^m : \frac{1}{(1-\xi^2)q} \leq w(S) <
      \frac{1}{(1-\xi^2)(q-1)}}$, for $2 \leq q \leq k-1$.
  \item $\R^m_{k} = \set{S \in \R^m : \frac{\alpha}{1-\delta} \leq w(S) <
      \frac{1}{(1-\xi^2)(k-1)}}$
  \item $\R^m_{k+1} = \set{S \in \R^m : \alpha \leq w(S) < \frac{\alpha}{1-\delta}}$
  \end{itemize}
  We partition $\D^m$ similarly: $\D^m_q = \set{(w(S),S): S \in \R_q}$ for
  $2 \leq q \leq k+1$.  

  Consider now the set $\V^m$
  of perfect valuation profiles 
  on $[k]$ and $[m']=\set{1,\ldots,m/2}$, extended to $[m]$ by 
  zero-extension.
  For a given $v \in \V^m$ and $2 \leq q \leq k,$ let us say that 
  $v \in \V^m_q$ if the set $\R^m_q$ contains an allocation $S$ 
  that achieves at least a $(1+\xi)(q-1)/k$ approximation to the social
  welfare maximizer.  Finally, let us say that $v \in \V^m_{k+1}$
  if $v$ does not belong to $\V^m_q$ for any $q < k+1.$  Notice that if 
  $v \not\in \V^m_q$ then the best approximation ratio achievable
  using an allocation in $\D^m_q$ is at most 
  \begin{equation} \label{eq:best-approx}
     \frac{(1+\xi)(q-1)}{k} \cdot \frac{1}{(1-\xi^2)(q-1)}
     = \frac{1}{(1-\xi)k} = \frac{1}{k} + \frac{\epsilon}{2} < \alpha.
  \end{equation}
  However, by our assumption that $\A$ achieves an $\alpha$-approximation
  for all valuation profiles with $m \in \FF$, the range $\D^m$ must
  contain an $\alpha$-approximation to the social welfare maximizer.
  If $m \in \FF$ and $v \in \V^m_{k+1}$, therefore, it follows that
  $\D^m_{k+1}$ must contain an $\alpha$-approximation to the social
  welfare maximizer.

  By the
  pigeonhole principle, at least one $q$ satisfies
  \begin{equation} \label{eq:vq}
    |\V^m_q| \geq \frac{1}{k} \cdot k^{m'}.
  \end{equation} 
  Let $T_q$ denote the set of all $m \in \FF$ such that \eqref{eq:vq} holds.
  By the preceding discussion, we have $\FF = \cup_{q=2}^{k+1} T_q.$  We now
  proceed to prove that $T_q$ is a PCD set for all $q$, completing the
  proof. 

 \paragraph{Cases 1 and 2: Large weight classes ($q \leq k$).}
 To each allocation $S$ of $m$ items to $k$ players, we may associate
 a function $f_S : [m'] \rightarrow [k] \cup \set{\ast},$ that maps
 each item $x \in [m']$ to the player who receives that item in $S$,
 or $\ast$ if the item is unallocated.  Similarly, to each perfect
 valuation profile $v$ on $[k]$ and $[m']$ we may associate a function
 $f_v:[m'] \rightarrow [k]$ that maps each item to the unique player
 who assigns a nonzero valuation to that item.  Note that $S$ achieves
 a $c$-approximation to the social-welfare-maximizing allocation for
 $v$ if and only if the functions $f_S$ and $f_v$ differ on $(1-c)m'$
 or fewer elements of $[m'].$

  Assume now that $q \leq k.$ If $m \in T_q$ then at least $1/k$
  fraction of all perfect valuation profiles in $\V^m_q$ have an
  allocation $S \in \R^m_q$ that achieves a
  $(1+\xi)(q-1)/k$-approximation to the maximum social welfare.  Thus,
  for at least $1/k$ fraction of all perfect valuation profiles $v \in
  \V^m_q$, there is some $S \in \R^m_q$ such that the $f_S$ and $f_v$
  differ on $\left(1 - \frac{q-1}{k} - \frac{(q-1)\xi}{k} \right) m'$
  elements of $[m']$.  Applying Lemma~\ref{lem:shatter-covering},
  there is a set $W$ of at least $\ceiling{\sigma m}$ elements of
  $[m']$, and a set $N'$ of $q$ players in $[k]$, such that all
  allocations of $W$ to $N'$ occur as restrictions of allocations in
  $\R^m_q.$ We refer to $W$ as a ``shattered'' subset of $[m']$.

  When $q<k$ (Case 1 of our argument) we may now construct, via a
  non-uniform polynomial-time reduction, an MIWR allocation rule for
  the $q$-player problem that achieves a $[(1-\xi^2)q]^{-1}$
  approximation for $\ceiling{\sigma m}$ items when $m \in T_q.$ Using
  $W$ and $N'$ -- as defined above -- as advice, embed the instance
  into an input for $\A$ by using players $N'$ and items $W$ in the
  obvious way: Give player in $[k] \sm N'$ an all-zero
  valuation. Moreover, extend the valuation of a player $i \in N'$ to
  the entire set of items $[m]$. Now, run $\A$ on the embedded
  instance. Notice that every allocation of $W$ to $N'$ appears as the
  restriction of some allocation in $R^m_q$, and is therefore in the
  range of $\A$ with weight at least $[(1-\xi^2)q]^{-1}$. Thus, $\A$
  must output a weighted allocation with expected welfare at least
  $[(1-\xi^2)q]^{-1}$ of the optimal. The result is a non-uniform
  poly-time MIWR mechanism for $q$ players with approximation ratio
  bounded away from $1/q$ for all integers $\hat{m}=\ceiling{\sigma
    m}$ such that $m \in T_q$.  By our induction hypothesis, this
  implies that the sum of all such $\hat{m}$ is a PCD set.  The fact
  that $T_q$ itself is a PCD set now follows as an easy application of
  the definition of PCD.

  When $q=k$ (Case 2 of our argument) using the same embedding yields
  an algorithm for $k$ players that achieves an $\alpha/(1-\delta)$
  approximation for all $\hat{m} = \ceiling{\sigma m}$ such that $m
  \in T_q.$ By our definition of $\alpha$, this implies that the set
  of all such $\hat{m}$ is a PCD set, which again implies that $T_q$
  is a PCD set.

  \paragraph{Case 3: The smallest weight class ($q=k+1$).}

  The remaining case is $q=k+1$.  When $m \in T_{k+1}$, by our
  definition of $\V^m_{k+1}$, at least $1/k$ fraction of all
  (extended) perfect valuation profiles $v \in \V^m$ have a weighted
  allocation $(w(S),S) \in \D^m_{k+1}$ that is an
  $\alpha$-approximation to the social welfare maximizing allocation
  for $v$.
  Since $\alpha \leq w(S) < \alpha/(1-\delta)$, the pure allocation
  $S$ must be a $(1-\delta)$-approximation to the social welfare
  maximizer.  On the other hand, our assumption is that maximizing
  social welfare is APX-hard, even for two players; to be specific,
  recall that $\eta>0$ was chosen such that it is NP-hard to
  approximate the maximum social welfare with approximation factor
  $1-\eta$.  We now complete the proof by exhibiting a randomized,
  non-uniform polynomial time algorithm that achieves a
  $(1-\eta)$-approximation for the $k$-player problem with $m/2$
  items, for all $m \in T_{k+1}$.  Notice that the de-randomization
  argument of Adleman \cite{Adl78} for proving $\compclass{BPP} \sse
  \compclass{P/Poly}$ can be used to de-randomize this to a
  non-uniform deterministic $(1-\eta)$-approximation for the
  $k$-player problem with $m/2$ items, for all $m \in T_{k+1}$. The
  reader unfamiliar with Adleman's argument may refer to Section
  \ref{subsec:mainresult}, where we use the argument to establish
  Theorem \ref{thm:main}.

Recall that 
$\delta= \epsilon \eta / 5k$. We will now use $\A$ to get a
$(1-\eta)$-approximate solution for an instance with $k$ players and
$m'=m/2$ items for all \comment{suffiently large?} $m \in T_{k+1}$
. We embed the instance on $k$ players and $m'$ items into $\A$ in the
following way.  We use $M_1 = [m] \sm [m']$ and let $v_i: 2^{M_1} \to
\RR$ denote the resulting valuation of player $i$.  We assume without
loss of generality that $\max_i v_i(M_1) = 1$.  Next, we modify each
player's valuation function by ``mixing in'' a perfect valuation
profile on the remaining set of items $M_2 = [m'].$ We draw a perfect
valuation profile $(v'_1,\ldots,v'_k)$ on $N$ and $M_2$ uniformly at
random. Now, we ``mix'' the original valuations $v$ with $v'$, in
proportions $1$ and $\gamma = \frac{4k}{\epsilon m'}$, to yield the
following \emph{hybrid valuation profile} $v^*: 2^M \to \RRp$.
\[v^*_i = v_i \oplus \gamma v'_i\]
We abuse notation and use $v_i$ [$v_i'$] to refer also to the
zero-extension of $v_i$ [$v_i'$] to $M$. Let $OPT=\max_{S \in \X}
v(S)$.  Similarly, let $OPT'=\max_{S \in \X} v'(S)$ and let
$OPT^*=\max_{S \in X} v^*(S)$.  Notice that $1 \leq OPT \leq k$, and
that $OPT'=m'$, by construction. Since $v$ and $v'$ are defined on a
disjoint set of items, it is easy to see that $OPT^* = OPT + \gamma
OPT'$. The scalar $\gamma$ was carefully chosen so that the following
facts hold:
\begin{enumerate}
\item The random valuation profile $v'$ accounts for a majority share
  of $v^*$ in any optimal solution. Specifically, $\gamma OPT' \geq
  \frac{4}{\epsilon} OPT$. This implies that an approximation to the
  optimal welfare using $v^*$ gives a similar approximation to the
  optimal welfare using $v'$. To be more precise, it can be shown by a
  simple calcluation that:

\begin{claim}\label{claim:approximates_perfect}
  For any $S \in \X$ and any $\beta \geq 0$, if $v^*(S) \geq \beta OPT^*$ then $v'(S)
  \geq (\beta - \epsilon/2) OPT'$.
\end{claim}

\item The original valuation profile $v$ accounts for a
  constant-factor share of $v^*$ in any optimal solution. Specifially
  $OPT \geq \frac{\epsilon}{4k} (\gamma OPT')$. This implies that a
  $(1-\delta)$-approximation to the optimal welfare using $v^*$ gives
  a $(1-O(\delta))$-approximation to the optimal welfare using $v$. To be
  more precise, it can be shown by a simple calculation that:
  \begin{claim}\label{claim:approximate_original}
  For any $S \in \X$, if $v^*(S) \geq (1-\delta) OPT^*$ then $v(S)
  \geq (1- \frac{5 k}{\epsilon} \delta) OPT = (1-\eta) OPT$.
 \end{claim}
\end{enumerate}

We are now ready to show that running $\A$ on the valuations $v^*$
will yield, with constant probability, an allocation that is a
$(1-\eta)$-approximation to the optimal welfare for the original
valuations $v$, when $m \in T_{k+1}$.
Let $(w(S),S)$
be the weighted allocation output by $\A$; note that $S$ is a random
variable over draws of $v'$. Since $\A$ is an $\alpha$ approximation
algorithm, the welfare $w(S) v^*(S)$ is at least $\alpha OPT^*
\geq(1/k + \epsilon) OPT^*$ with probability $1$. This implies that
$v^*(S) \geq \left(\frac{1}{w(S)\cdot k} + \frac{\epsilon}{w(S)}\right)
OPT^*$. By Claim \ref{claim:approximates_perfect}, we see that $v'(S)$
is not too far behind: $v'(S) \geq \left(\frac{1}{w(S)\cdot k} +
\frac{\epsilon}{w(S)} - \frac{\epsilon}{2}\right) OPT'$. Moreover, this
gives:
\begin{equation} \label{eq:vprime-approx}
 w(S) v'(S) \geq \left(\frac{1}{k} + \frac{\epsilon}{2}\right) OPT'
\end{equation}
Recall from equation \eqref{eq:best-approx} that if $v' \in \V^m_{k+1}$
then for $2 \leq q \leq k,$ there is no $S \in \R^m_q$ that satisfies 
\eqref{eq:vprime-approx}, hence any such $S$ satisfying
\eqref{eq:vprime-approx} must belong to $\D^m_{k+1}.$  
Also, by our assumption that $m \in T_{k+1}$, the probability
that $v' \in V^m_{k+1}$ is at least $1/k$.  

We have thus established that running $\A$ on the random input
$v^*$ yields, with probability at least $1/k$, an outcome
$(w(S),S)$ in $\D^m_{k+1}$.  Using the fact that 
$w \leq \alpha/(1-\delta)$ and $w(S) v^*(S) \geq \alpha OPT$, 
we conclude that $S$ is
$(1-\delta)$-approximate for $v^*$ also with probability $1/k$:
\[ v^*(S) \geq (1-\delta) OPT^* \]
Invoking Claim \ref{claim:approximate_original}, we conclude that
$v(S) \geq (1-\eta) OPT$ with constant probability over draws of
$v'$. Since $w(S)$ is at least $1/k$, $S$ is
output by $\A$ with constant probability. This completes the proof.

\subsection{Main Result}
\label{subsec:mainresult}

In this section, we complete the proof of Theorem
\ref{thm:main}. First, we make the observation that running a
randomized MIWR algorithm multiple times independently and returning the
best allocation output by any of the runs results in another
randomized MIWR algorithm.

\begin{lemma}\label{lem:miwr_closed}
  Fix a randomized MIWR algorithm $\A$ and a positive integer $r$. 
  Let $\A^r$ be the algorithm
  that runs $r$ independent executions of $\A$  on its input, 
  and of the $r$ allocations
  returned, outputs the one with greatest welfare. $\A^r$ is also
  randomized MIWR.
\end{lemma}
\begin{proof}
  Condition on $\D_1, \ldots, \D_r$, the ranges of $\A$ on the $r$
  independent executions.  $\A$ maximizes expected welfare over
  $\D_i$ on execution $i$. Therefore $\A^r$ maximizes over 
  $\D_1 \union \cdots \union \D_r$.
\end{proof}

Now, we derive Theorem \ref{thm:main} from Theorem
\ref{thm:main_poly}, using a de-randomization argument similar to that
of Adleman \cite{Adl78}. Assume for a contradiction that $\A$ is a
randomized MIWR algorithm that runs in polynomial time and achieves an
expected approximation ratio $1/k + \epsilon$ for each input $m$ and
$v_1,\ldots,v_k$.  Let $n$ denote the number of bits in the input, and
let $\ell(n)$ be a polynomial bounding the length of the random string
drawn by $\A$. We will describe a polynomial-time with
polynomial-advice MIWR algorithm that achieves an approximation ratio
of $1/k + \epsilon / 2$, which contradicts
Theorem~\ref{thm:main_poly}.

Let $r(n)=2 n /\epsilon^2$ and let $\A' = \A^{r(n)}$.  By Lemma
\ref{lem:miwr_closed}, $\A'$ is randomized MIWR, runs in polynomial
time, and draws at most $\ell(n) r(n)$ random bits. Let $X_i$ be the
fraction of the optimal social welfare achieved by the allocation
output on the $i$'th run of $\A$. The random variables
$X_1,\ldots,X_{r(n)}$ are independent, $0 \le X_i \le 1$, and $E[X_i]
\ge 1/k + \epsilon$. For each input of length $n$, the probability
that none of the $r(n)$ runs of~$\A$ return an allocation with welfare
better than $1/k + \epsilon/2$ of the optimal can be upper-bounded
using Hoeffding's inequality:

\begin{equation*}\label{eq:Hoeffding}Pr \left[\max_i X_i \le \left({\frac{1}{k}}
+ {\frac{\epsilon}{2}} \right)\right] \le Pr \left[E\left(\sum_i X_i\right) -
\sum_i X_i \ge
\frac{\epsilon r(n)}{2} \right] \le e^{-\epsilon^2 r(n) / 2} = e^{-n}.
\end{equation*}

The number of different inputs of length $n$ is $2^n$. Thus, using the
union bound and the above inequality, the probability that $\A$
outputs a $(1/k + \epsilon/2)$-approximate allocation on all
inputs of length $n$ is non-zero. Therefore, for each $n$ there is choice of at
most $\ell(n) r(n)$ random bits such that $\A'$ achieves a $1/k +
\epsilon/2$ approximation for all inputs. Using this as the advice
string, this contradicts Theorem \ref{thm:main_poly}. This completes
the proof of Theorem \ref{thm:main}.



\section{Hardness Result for Budgeted Additive Valuations}
\label{sec:MIR-BA}

In this section, we prove the following theorem:

\begin{theorem}\label{thm:randomized_MIR-BA}
There is no polynomial time randomized MIR mechanism that achieves
$\frac{1}{k} + \epsilon$ approximation of the optimal social welfare
for $k$~bidders with budgeted additive valuations, unless
$\compclass{NP} \subseteq \compclass{P/Poly}$.
\end{theorem}

Notice that this theorem implies that all universally truthful
randomized MIR mechanisms cannot achieve $\frac{1}{k} + \epsilon$
approximation for $k$~bidders. In the proof we will use the term
\emph{$k$-partition} interchangeably with an allocation for
$k$~bidders. As in the previous section, an allocation does not
necessarily allocate all items. A partition corresponding to a total
allocation will be called \emph{covering}. Perfect valuations
generated by total allocations will also be used in the proof. We
first study the abundance of ``orthogonal'' partitions of~$M=[m]$. The
following definition formalizes this notion.






\begin{definition} Let $\mathcal{T}$ be a set of
$k$-partitions of~$M$: $\mathcal{T} = \{(T_1^i, T_2^i, \cdots T_k^i)
\: | \: i \in [\ell] \}$, we say these partitions are
\emph{$\epsilon$-apart} for $\epsilon > 0$, if,
$$\forall (I_1, I_2, \cdots, I_\ell) \in [k]^\ell, \;\;  \left|\bigcap_{i =
1}^\ell T^i_{I_{i}} \right| \le \left(\frac{1}{k}\right)^\ell (1 +
\epsilon) |M|.$$
\end{definition}


\begin{lemma}\label{lemma:largeF}
  For every pair of integers $k$ and $\ell$, and every $\epsilon$
  satisfying $0< \epsilon < 1$, there exists $\alpha > 0$ such that
  there exists a set~$F$ of covering $k$-partitions of~$M$, where $|F|
  = e^{\alpha m}$, and every $\ell$ elements of~$F$ are
  $\epsilon$-apart.
\end{lemma}

The proof of Lemma~\ref{lemma:largeF} is relegated to the appendix.
The next lemma shows that for valuations generated by partitions
that are apart, good welfare approximations require distinct
allocations.

\begin{lemma}\label{lemma:smallwelfare}
Given $\ell$ covering $k$-partitions that are $\epsilon$-apart,
consider the $\ell$ tuples of valuations generated by them. The sum
of social welfare achievable by a single allocation on these tuples
of valuations is at most $\left(\frac{1}{k} +
\sqrt{\frac{k\pi}{\ell}} + \frac{k 2^{-\ell/k}}{\ell \cdot \ln
2}\right)(1 + \epsilon)$ of the optimal.
\end{lemma}

\begin{proof}
The sum of optimal social welfare for $\ell$ generated valuations is
easily seen to be $\ell m$. Let $\mathcal{T} = \{(T_1^i, T_2^i,
\ldots, T_k^i)\: | \: i \in [\ell] \}$ be the set of covering
$k$-partitions that generate the valuations. Let $S$ be the sum of
social welfare achieved by one single allocation~$R$, then $S$ is
maximized when $R(x) \in \underset{i \in [k]}{\mathrm{argmax}}
\sum_{j = 1}^\ell v_i^j(x)$ holds for every item $x \in M$. For $I
\in [k]^\ell$, define $Q_i(I)$ to be the number of times that $i$
occurs in~$I$, i.e., $|\{j \in [\ell] \:|\: I_j = i\}|$, and define
the \emph{plenty} of~$I$ to be $P(I) = \max_{i \in [k]} Q_i(I).$
Then we have
\begin{equation}\label{eq:boundS}S \le \sum_{I \in [k]^\ell} P(I)
\cdot |T^1_{I_1} \cap T^2_{I_2} \cap \cdots \cap T^\ell_{I_\ell}|
\le \frac{\sum_{I \in [k]^\ell} P(I) }{k^\ell}(1 +
\epsilon)m.\end{equation} The second inequality results from the
$\epsilon$-apartness of the partitions.

We recognize that $\frac{\sum_{I \in [k]^\ell} P(I) }{k^\ell}$ can
be seen as the expectation of a properly defined random variable
--- if $I'$ is a random variable uniformly distributed on
$[k]^\ell$, then this factor is exactly the expectation of $P(I')$.
The problem boils down to bounding $E[P(I')]$. Note that $E[Q_i(I')]
= \ell / k$. Let $Y(I) = P(I) - \frac{\ell}{k}$, then by the union
bound
$$Pr[Y(I') > \delta \cdot \frac{\ell}{k}] \le k \cdot Pr [Q_i(I') > (1 + \delta) \frac{\ell}{k}].$$
Applying Chernoff bound (Theorem~\ref{thm:Chernoff}), we get
\begin{eqnarray*}
E[Y(I')] &=& \int_0^\infty Pr_{I'}[Y(I') > \delta] \:
\mathrm{d}\delta
 = \frac{\ell}{k} \int_0^\infty Pr_{I'} [Y(I') >
 \left(\delta \cdot \frac{k}{\ell} \right) \frac{\ell}{k}]
 \: \mathrm{d}\left(\delta \cdot \frac{k}{\ell} \right) \\
&\le& \ell \left\{\int_0^{2e-1} Pr_{I'} [Q_i(I') > (1 + \delta)
\frac{\ell}{k}] \: \mathrm{d}\delta  + \int_{2e-1}^\infty Pr_{I'}
[Q_i(I') > (1 + \delta)
\frac{\ell}{k}] \mathrm{d}\delta  \right\} \\
&<& \ell \left(\int_0^\infty e^{-\delta^2 \ell/4k} \:
\mathrm{d}\delta + \int_1^\infty 2^{-\delta \ell/k} \:
\mathrm{d}\delta \right) = \sqrt{k\pi \ell} + \frac{k
2^{-\ell/k}}{\ln 2}.
\end{eqnarray*}
Consequently we also obtain that $E[P(I')] \le \frac{\ell}{k} +
\sqrt{k\pi \ell} + \frac{k 2^{-\ell/k}}{\ln 2}$. Substituting this
into \eqref{eq:boundS}, then dividing it by $\ell m$, the sum of
optimal social welfare, we get what the lemma claims.
\end{proof}

To finish the proof, we need the next lemma that connects to the
``shattering'' lemma.

\begin{definition}\label{def:epsilonfar}
Two k-partitions $(T_1, T_2, \cdots, T_k)$ and~$(T_1', T_2', \cdots,
T_k')$ are said to be \emph{$\epsilon$-far} if $\sum_{i \neq j} |T_i
\cap T_j'| \ge \epsilon m$. If two partitions are not
$\epsilon$-far, we say that they are \emph{$\epsilon$-close}.
\end{definition}

\begin{lemma}\label{lemma:shatter_BAMIR}
If an $MIR$ mechanism achieves at least $\frac{1}{k} +
\sqrt{\frac{k\pi}{\ell}} + \frac{k e^{-\ell/k}}{\ell \cdot \ln 2} +
\epsilon'$ approximation to the optimal social welfare for $k$
bidders, where $k$, $\ell$ and $\epsilon'$ are all fixed, and $\ell
\ge 10k$, then there is a $\delta > 0$, such that there is a subset
$S$ of items, with $|S| \ge \delta m$, and two bidders $i$ and~$j$,
and every allocation of items in~$S$ to $i$ and~$j$ is a restriction
of an allocation in the range of the mechanism.
\end{lemma}

\begin{proof}
Let $\epsilon$ be $\epsilon'/3$, and by Lemma~\ref{lemma:largeF}
there is a set~$F$ of $\epsilon$-apart $k$-partitions, and $|F| =
e^{\alpha|U|}$ for some $\alpha > 0$. Let $\mathcal{R}$ be the
multi-set of allocations output by the mechanism on the tuples of
valuations generated by the partitions in~$F$. Note that
$|\mathcal{R}| = |F|$. We claim that in~$\mathcal{R}$, there can be
no $\ell$ partitions such that every two of them are
$\frac{\epsilon}{\ell^2}$-close. For a contradiction, suppose that
this is the case. Let $\{(T_1^i, T_2^i, \cdots, T_k^i) \: | \: i \in
[\ell]\}$ be these allocations, define $D$ to be $\bigcup_{1 \le s <
t \le \ell} \bigcup_{i \neq j} (T_i^s \cap T_j^t)$, then because of
the pairwise $\frac{\epsilon}{\ell^2}$-closeness of the partitions,
$|D| \le \epsilon m$. For each item not in~$D$, the $\ell$
allocations either allocate it in the same way, or some allocate it
in the same way and others do not allocate it to any bidder. By
Lemma~\ref{lemma:smallwelfare},  on $M \backslash D$ the allocations
can achieve at most $\frac{1}{k} + \sqrt{\frac{k\pi}{\ell}} +
\frac{k 2^{-\ell/k}}{\ell \cdot \ln 2} + \epsilon$ of the sum of
optimal welfare. Each item in~$D$ can contribute at most $\ell$ to
the sum of welfare, and in total they count at most $\epsilon$
fraction of the optimal, which is $\ell m$. Thus the mechanism can
achieve at most $\frac{1}{k} + \sqrt{\frac{k\pi}{\ell}} + \frac{k
2^{-\ell/k}}{\ell \cdot \ln 2} + \frac{2\epsilon'}{3}$ of the
optimal social welfare, contradicting the assumption on its
performance. Thus, there are no $\ell$-allocations in~$\mathcal{R}$
that are pairwise $\frac{\epsilon}{\ell^2}$-far. Applying
Lemma~\ref{lem:shatter-eps-far} to~$\mathcal{R}$, we finish the
proof.
\end{proof}

\begin{prevproof}{Theorem}{thm:randomized_MIR-BA}: Whenever we have the range of a mechanism
containing all allocations of items in a linearly smaller subset to
two bidders, we can use the mechanism with polynomial advice to
optimize the social welfare of an auction with fewer items.
Therefore, the condition of Lemma~\ref{lemma:shatter_BAMIR} should
not be satisfied for any~$\ell$ unless $\compclass{NP} \subseteq \compclass{P/Poly}$. Let
$\ell$ in Lemma~\ref{lemma:shatter_BAMIR} get arbitrarily big. We
see that, unless $\compclass{NP} \subseteq \compclass{P/Poly}$, any efficient MIR mechanism
cannot achieve $1/k + \epsilon$ approximation to the optimal social
welfare for $k$~bidders. Then by the same argument as in
Section~\ref{subsec:mainresult} (proof omitted here), we can extend
this to randomized MIR mechanisms and get
Theorem~\ref{thm:randomized_MIR-BA}.

\end{prevproof}




\section*{Acknowledgements}

We thank the authors of~\cite{BU09} and~\cite{MPSS09} --- Elchanan Mossel, Christos Papadimitriou, Michael Schapira, Yaron Singer, Dave Buchfuhrer, and Chris Umans --- for stimulating discussions on these topics and for their influence on this work.

\bibliographystyle{plain}
\bibliography{amplified}

\appendix

\section{Shattering Results}
\label{app:shatter}
\newcommand{\Ham}{{\operatorname{Ham}}}

We first formally define the notion of ``shattering'' in a more
general setting.

\begin{definition}
For any sets $U, \, V$ we interpret the notation $V^U$ to mean the
set of functions from $U$ to $V$. If $R \subseteq V^U, \, S
\subseteq U, \, L \subseteq V$, we say that $S$ is $(L,
q)$-shattered by $R$, for an integer $q$, $2 \le q \le |L|$, if
there exist $q$~functions $c_1,c_2, \ldots, c_q : S \rightarrow L$
that satisfy:
\begin{enumerate}
\item $\forall x \in S \;\; \forall i \neq j \;\; c_i(x) \neq c_j(x)$
\item $\forall h \in [q]^S \;\;
\exists f \in R \;\; \forall x \in S \;\; f(x) = c_{h(x)}(x)$
\end{enumerate}
\end{definition}

Intuitively, we associate with each element in~$S$ a range in~$L$ of
size exactly~$q$, and we say that $S$ is $(L, q)$-shattered by~$R$
if every function that maps each element in~$S$ to its associated
range is a restriction of an element in~$R$. In the context of
combinatorial auctions, we see $U$ as the set of items, and $V$ as
the set of bidders, plus a dummy bidder representing not allocating
the item. Then set of functions $V^U$ is the set of all possible
allocations.

The following observation bridges this notion of shattering to its
application to the combinatorial auctions in the paper.

\begin{observation} \label{obs:strongshatter}
If a subset~$S$ of size $\delta m$ is $(L, q)$-shattered by $R
\subseteq V^U$, then there exists a subset $L' \subseteq L$ and $S'
\subseteq S$, such that $|L'| = q$, $|S'| \ge |S| / {|L| \choose q}$
and $S'$ is $(L', q)$-shattered by~$R$.
\end{observation}

The observation is easily seen by the pigeonhole principle. Note
that by the definition of $(L, q)$-shattering, if $|L'| = q$, then
we have that every function from $S'$ to $L'$ is a restriction of an
element in~$R$. In the context of combinatorial auctions, this means
that all possible allocations of items in~$S'$ to the $q$~bidders
in~$L'$ are in the range~$R$ under restriction. It is this form of
``strong'' shattering that is in use in the main body of the paper.
In the following lemmas, we will show the existence of large $(L,
q)$-shattered sets, being aware that an application of the above
observation implies a subset being ``strongly'' shattered, of size
only a constant factor smaller.

\begin{lemma} \label{lem:shatter1}
For every integer $n \geq 2$, $q$, $2 \leq q \leq n$, and every
$\epsilon
> 0$, there is a $\delta>0$ such that the following holds. For every
pair of finite sets $M,N$ with $|N|=n$ and every set $R$ of more
than $(q - 1 + \epsilon)^ {|M|}$ elements of $N^M$ there is a set
$S$ of at least $\delta |M|$ elements of $M$ such that $S$ is $(V,
q)$-shattered by $R$.
\end{lemma}
\begin{proof}
Let $F_q(m,n,d)$ denote the maximum cardinality of a set $R
\subseteq A^B$ such that $|A| = n, \, |B| = m,$ and $R$ does not
$(A, q)$-shatter any $(d+1)$-element subset of $B$.

Fix an element $b \in B$.  For each element $f \in R$, let $f_{-b}$
denote the restriction of~$f$ to the set~$B\backslash\{b\}$.  Take
the set of all functions $g: B\backslash\{b\} \to A$ and partition
it into sets $Q_0, Q_1, \cdots, Q_{n \choose q}$ as follows. First,
given an ordered pair $(g,a)$ consisting of a function $g$ from
$B\backslash\{b\}$ to $A$ and an element $a \in A$, let $g*a$ denote
the unique function $f$ from $B$ to $A$ that maps $b$ to $a$ and
restricts to $g$ on $B\backslash\{b\}$. Now define $S(g)$ to be the
set of all $a \in A$ such that $g*a$ is in $R$. Number all the
$q$-element subsets of $A$ from $1$ to ${n \choose q}$, call them
$P_1, P_2, \cdots, P_{n \choose q}$, and let $Q_i$ ($1 \leq i \leq
{n \choose q}$) consist of all $g$ such that $S(g)$ has at least
$q$~elements, and the $q$ smallest elements of $S(g)$
constitute~$P_i$. Finally let $Q_0$ consist of all $g$ such that
$S(g)$ has fewer than $q$~elements.

By our assumption that $R$ does not $(A,q)$-shatter any set of size
greater than~$d$, we have the following facts:
\begin{enumerate}
\item $Q_0$ does not $(A, q)$-shatter any $(d+1)$-element
subset of $B \setminus \{b\}.$  Consequently, $$|Q_0| \leq
F(m-1,n,d).$$

\item For all $i \leq {n \choose q}$, $Q_i$ does not $(A, q)$-shatter
any $d$-element subset of $B \setminus \{b\}.$  Consequently,
$$|Q_i| \leq F_q(m-1,n,d-1).$$
\end{enumerate}

Let $R_i$ denote the set of all $f \in R$ such that $f_{-b}$ is in
$Q_i$, for $0 \leq i \leq {n \choose q}$, then by definition of
$Q_i$, we have $|R_0| \le (q - 1) |Q_0|$, and $|R_i| \le n|Q_i|$ for
$i \le 1$. Since $R_i$'s are disjoint, we have
$$ |R| = \sum_{i = 0}^{n \choose q} |R_i| \le (q-1)|Q_0| + \sum_{i = 1}^{n \choose q} n |Q_i|,$$

\begin{equation}\label {eq:recur} F_q(m,n,d) \le (q-1) F_q(m-1,n,d) + n{n
\choose q} F_q(m-1,n,d-1)
\end{equation}

The recurrence \eqref{eq:recur}, together with the initial condition
$F_q(m,n,0)=(q-1)^m$ for all $m,n$, implies the upper bound
\[
F_q(m,n,d) \leq \sum_{i=0}^d n^i \binom{n}{q}^i \binom{m}{i} (q -
1)^m
\]
Thus, if $F_q(m,n,d) > (q - 1 + \epsilon)^{m}$ then, by using
Stirling's approximation, we see that $d > \delta m$ for some $\delta$
depending only on $\epsilon$ and $n$.
\end{proof}

In Section~\ref{sec:midr_apx} of the paper, we made use of the fact that a
range of allocations shatters a large subset if they generate good
social welfare for many perfect valuations. The condition is
captured by the following definition:

\begin{definition}
For two functions $f,g \in N^M$, their \emph{normalized Hamming
distance} $\Ham(f,g)$ is equal to $\frac{1}{|M|}$ times the number
of distinct $x \in M$ such that $f(x) \neq g(x).$  If $f \in N^M$
and $R \subseteq N^M$, the Hamming distance $\Ham(f,R)$ is the
minimum of $\Ham(f,g)$ for all $g \in R$.
\end{definition}

As each perfect valuation can be seen as a function~$f$ in~$N^M$,
and each allocation can be viewed as a $g \in N^M$, $\Ham(f, g)$ is
how much social welfare is lost by~$g$ on the perfect valuation~$f$.
In the same way, $R$ can be viewed as a range of allocations, and
$\Ham(f, R)$ is the minimum social welfare lost by any of the
allocation in~$R$ on valuation~$f$. If $\Ham(f, R)$ is small for a
large fraction of $f \in N^M$, it means the range achieves a good
approximation of social welfare for a significant portion of the
perfect valuations.

We also note that since $N$ can represent the set of bidders plus a
dummy bidder representing not allocating an item, $N^M$ can express
all allocations including those not allocating all items. On the
other hand, if we restrict the functions so that they can take
values only in a subset $L$ representing the real bidders, then they
represent allocations that do not discard items. This explains the
role played by the set~$L$ in the next lemma.

\begin{lemma}  \label{lem:shatter2}
For every positive real number $\epsilon > 0$, integers $n \geq 2$,
$q$, $2 \le q \le n$, and polynomial $\gamma(n)$, there is a $\delta
> 0$ such that the following holds. For all finite sets $M,N$ and
all subsets $L \subseteq N$ with $|L|=n$, if $R \subseteq N^M$ and
at least $\gamma n^{|U|}$ points $f \in L^U$ satisfy $\Ham(f,R) < 1
- (q - 1)/n - \epsilon,$ then there is a set $S \subseteq M$ such
that $|S|
> \delta |M|$ and $S$ is $(L, q)$-shattered by $R$.
\end{lemma}
\begin{proof}
Let $m=|M|, r = 1 - (q - 1)/n - \epsilon.$ Let $A$ be the set of all
points $f \in L^M$ such that $\Ham(f,R) < r.$  Let $G$ be a function
from $A$ to $R$ such that $\Ham(f,G(f)) < r$ for all $f \in A$. Let
$I(f)$ denote the set of all $x \in M$ such that $f(x) = G(f)(x)$.
Note that our assumption that $\Ham(f,G(f)) < r$ implies that
$|I(f)| \geq (\frac{q-1}{n}+\epsilon)m$.  The number of pairs
$(f,J)$ such that $f \in A, \, |J| = \epsilon m / 2, \, J \subseteq
I(f)$ is bounded below by $\gamma n^m \cdot
\binom{(1/n+\epsilon)m}{\epsilon m /2}$.
By the pigeonhole principle, there is at least one set $J$ of
$\epsilon m / 2$ elements such that the number of $f \in L^U$
satisfying $J \subseteq I(f)$ is at least
\begin{align*}
\gamma n^m \cdot \left. \binom{(\frac{q-1}{n}+\epsilon)m}{\epsilon m
/2} \right/ \binom{m}{\epsilon m / 2} &= \gamma n^m \frac{
((\frac{q-1}{n} + \epsilon) m)! \; ((1 - \epsilon/2)m)! }
{ ((\frac{q-1}{n} + \epsilon/2)m)! \; m! } \\
& > \gamma n^m \cdot \frac{(\frac{q-1}{n} + \epsilon)m}{m} \cdot
\frac{(\frac{q-1}{n}+\epsilon)m-1}{m-1} \cdots
\frac{(\frac{q-1}{n}+\epsilon/2)m}{(1-\epsilon/2)m} \\
& > \gamma n^m \left( \frac{\frac{q-1}{n} + \epsilon/2}{1 -
\epsilon/2} \right)^{\epsilon m / 2}.
\end{align*}
Fix such a set $J$.  For every $f \in L^M$ satisfying $J \subseteq
I(f)$, the restriction of $f$ to $J$ is an element $g \in L^J;$ note
that $g$ is also the restriction of $G(f)$ to $J$.  For any single
$g \in L^J$, the number of $f \in L^M$ that restrict to $g$ is
bounded above by $n^{m-\epsilon m / 2}$.  Applying the pigeonhole
principle again, we see that the number of distinct $g \in L^J$ that
occur as the restriction of some $f \in A$ satisfying $J \subseteq
I(f)$ must be at least
\begin{align*}
\left. \gamma n^m \left( \frac{\frac{q-1}{n} + \epsilon/2}{1 -
\epsilon/2} \right)^{\epsilon m / 2} \right/ n^{m - \epsilon m / 2}
&= \gamma \left( \frac{q - 1 + \epsilon n / 2}{1 - \epsilon / 2}
\right)^{\epsilon m / 2}.
\end{align*}
We now have the following situation.  There is a set $J$ of
$\epsilon m / 2$ elements, and at least $\gamma \cdot (q - 1 +
\epsilon n / 2)^{|J|}$ elements of $L^J$ occur as the restriction of
an element of $R$ to $J$.  It follows from Lemma~\ref{lem:shatter1}
that $J$ has a subset of $S$ of at least $\delta m$ elements such
that $S$ is $(L, q)$-shattered by $R$.
\end{proof}

\begin{prevproof}{Lemma}{lem:shatter-covering}
Combining Lemma~\ref{lem:shatter1}, Lemma~\ref{lem:shatter2} and
Observation~\ref{obs:strongshatter}, we immediately get
Lemma~\ref{lem:shatter-covering}.
\end{prevproof}

Now we are ready to show the next shattering lemma. The
$k$-partitions represent allocations to $k$-bidders, allowed to
discard items. In the proof we occasionally see them as partial
functions, in a way very similar to that in the previous lemmas.

To state the lemma succinctly, we denote by $t(k, \ell, \epsilon, m,
r)$ the smallest number of subsets that are $([k], 2)$-shattered by
any set~$\mathcal{T}$ of $k$-partitions of a set~$M$, where $|M| =
m$, $|\mathcal{T}| = r$, and every $\ell$ partitions
from~$\mathcal{T}$ are not pairwise $\epsilon$-close (see
Definition~\ref{def:epsilonfar}).

\begin{lemma}\label{lemma:shatter3}
For every integers $k \ge 2$, $\ell$ and~$m$, every $\epsilon > 0$,
there exists an $\alpha > 0$ such that $t(k, \ell, \epsilon, m, r)
\ge r^\alpha$ for every~$r$.
\end{lemma}

Note that if $r$ is $2^{\beta m}$ for some $\beta > 0$, the
conclusion implies the existence of a subset of size $\gamma m$ that
is $([k], 2)$-shattered by the set of partitions, for some $\gamma
> 0$. The proof also works if $R$ is a multi-set.

\begin{proof}
Let $\mathcal{T}$ be any set of $k$-partitions of~$M$, $|M| = m$,
$|\mathcal{T}| = p$, and every $\ell$ elements from~$\mathcal{T}$
are not pairwise $\epsilon$-close. We arbitrarily group the
partitions in~$\mathcal{T}$, so that every group consists of $\ell$
partitions. Then in each group $\{(T_1^i, T_2^i, \cdots, T_k^i) \: |
\: i \in [\ell]\}$, there are at least two partitions that are
$\epsilon$-far, and the size of their ``difference'' $\sum_{i \neq
j} |T_i^s \cap T_j^t|$ is at least $\epsilon m$. Since we have
$r/\ell$ such pairs, the sum of the sizes of ``differences'' will be
at least $\frac{\epsilon m r}{\ell}$. By pigeonhole principle, there
exists an $x \in M$, and $i^*, j^* \in [k]$ ($i^* \neq j^*$) such
that in at least $\frac{\epsilon r}{\ell \binom{k}{2}}$ pairs of
partitions $(T_1, T_2, \cdots, T_k)$ and $(T_1', T_2', \cdots,
T_k')$, $x$ occurs in $(T_{i^*} \cap T_{j^*}') \cup (T_{i^*}' \cap
T_{j^*})$. Now if we denote by $\mathcal{T}_{i^*}$ the set of those
partitions in~$\mathcal{T}$ that map $x$ to~$i^*$, and
$\mathcal{T}_{j^*}$ those mapping $x$ to~$j^*$, then
$|\mathcal{T}_{i^*}| \ge \frac{\epsilon p}{\ell \binom{k}{2}}$ and
$|\mathcal{T}_{j^*}| \ge \frac{\epsilon p}{\ell \binom{k}{2}}$.

Let $I_{i^*}$ denote the set of subsets that are $([k],
2)$-shattered by~$\mathcal{T}_{i^*}$, and similarly $I_{j^*}$ the
set of subsets $([k], 2)$-shattered by~$\mathcal{T}_{j^*}$, $I$ the
set of subsets $([k], 2)$-shattered by~$\mathcal{T}$ itself. We
claim that $|I| \ge |I_{i^*}| + |I_{j^*}|$. To see this, it is clear
that $I_{i^*} \cup I_{j^*} \subseteq I$, and $x$ is not in any of
the set in $I_{i^*} \cup I_{j^*}$. Besides, for every set~$S$ in
$I_{i^*} \cap I_{j^*}$, $S \cup \{x\}$ should be shattered
by~$\mathcal{T}$ according to our definition. Therefore $|I| \ge
|I_{i^*} \cup I_{j^*}| + |I_{i^*} \cap I_{j^*}| = |I_{i^*}| +
|I_{j^*}|$. In other words, $t (k, \epsilon, m, r) \ge 2 t(k,
\epsilon, m, \frac{\epsilon r}{\ell \binom{k}{2}})$. By induction
the lemma is proved.
\end{proof}

\begin{prevproof}{Lemma}{lem:shatter-eps-far}
Combining Lemma~\ref{lemma:shatter3} and
Observation~\ref{obs:strongshatter}, we get
Lemma~\ref{lem:shatter-eps-far}.
\end{prevproof}


\include{app_midrapx}
\section{Omitted Proofs from Section~\ref{sec:MIR-BA}}
\label{sec:app_MIRBA}

\begin{prevproof}{Lemma}{lemma:largeF} This is shown by a
probabilistic argument. We randomly sample a number of covering
$k$-partitions in the following way. Each time we sample a
partition, we decide for each item in~$M$, uniformly at random,
which one of the $k$~subsets it should be placed in. We repeat this
process $n$~times, and get a set of $k$-partitions $\{(T_1^i, T_2^i,
\cdots, T_k^i)\: | \: i \in [n]\}$. Let $A = \{a_1, a_2, \ldots,
a_\ell\} \subseteq [n]$, $B = \{b_1, b_2, \cdots, b_\ell\} \in
[k]^\ell$, let $I_A^B$ denote the event
\begin{equation} \label{eq:I_A^B-def}|T_{b_1}^{a_1} \cap T_{b_2}^{a_2} \cap \cdots \cap
T_{b_\ell}^{a_{\ell}}| > \frac{1}{k^\ell}(1+\epsilon)m.
\end{equation}
The expectation of the left hand side of~\eqref{eq:I_A^B-def}
is~$m/k^\ell$. By Chernoff bound,
$$Pr[I^B_A] \le e^{-\epsilon^2 m/4k^\ell}, \;\; \forall A, B.$$

The probability that $I^B_A$ happens for some~$A$ and~$B$ is upper
bounded by
$$\sum_{A \subseteq [n], |A| = \ell} \sum_{B \in [k]^\ell} Pr[I^B_A]
\le \binom{n}{\ell} k^\ell e^{-\epsilon^2 m /4 k^\ell} \le \left
(\frac{ken}{\ell} \right)^\ell e^{-\epsilon^2 m /4 k^\ell}.$$
Therefore as long as $n < \frac{\ell}{ke} e^{\epsilon^2 m/4 \ell
k^\ell}$, the probability above is smaller than~1, i.e., there exist
$n$~partitions that satisfy the lemma. This completes the proof.
\end{prevproof}

In the proof of Lemma~\ref{lemma:smallwelfare}, we used two forms of
the Chernoff bound:

\begin{theorem} \label{thm:Chernoff} (Chernoff bound): Let $X_1, X_2, \cdots, X_n$ be i.\@i.\@d.\@
random variables such that $X_i \in \{0, 1\}$ and $Pr(X_i = 1) = p$
for every~$i \in [n]$. Let $X = \sum_i X_i$ and $\mu = E[X]$, then

(i) for $\delta > 2e - 1$, $Pr[X > (1 + \delta)\mu] < 2^{-\delta
\mu}$;

(ii) for $0 < \delta < 2e - 1$, $Pr[X > (1 + \delta) \mu] <
e^{-\delta^2 \mu / 4}$.
\end{theorem}


\section{Omitted Proofs from Section~\ref{sec:complexity}}
\label{app:complexity}

\begin{prevproof}{Lemma}{lem:finite-union}  Suppose $S_1,\ldots,S_k$ are CD sets, with circuit
families $\{\circuit_n^{(i)}\}$ $(1 \leq i \leq k)$ such that 
$\circuit_n^{(i)}$ has size bounded by a polynomial $q_i(n)$
and decides {\sc 3sat} correctly on all instances of size $n \in S_i.$
Let $q(n)$ be a polynomial satisfying $q(n) \geq 
\max_{1 \leq i \leq k} q_i(n)$ for all $n \in \NN.$  
We can obtain a family of circuits $\{\circuit_n\}$ of size
bounded by $q(n)$, by defining $\circuit_n$ to be equal to
$\circuit_n^{(i)}$ if $n$ belongs to $S_i$ but not to
$S_1,\ldots,S_{i-1}$, and defining $\circuit_n$ to be
arbitrary if $n \not\in S_1 \cup \cdots \cup S_k.$
Then $\circuit_n$ decides {\sc 3sat} correctly on all
instances of size $n \in S_1 \cup \cdots \cup S_k$, as
desired.

If $S_1,\ldots,S_k$ are CD sets, $p_1,\ldots,p_k$ are polynomials,
and for $1 \leq i \leq k$ we have a PCD set $T_i \subseteq
\bigcup_{n \in S_i} [n,p_i(n)],$ then we may take $p(n)$ to be
any polynomial satisfying $p(n) \geq \max_{1 \leq i \leq k} p_i(n)$
for all $n \in \NN$, and we may take $S$ to be the set
$S_1 \cup \cdots \cup S_k.$  Then we find that the set 
$T = T_1 \cup \cdots \cup T_k$
is contained in $\bigcup_{n \in S} [n,p(n)].$
This implies that $T$ is PCD, because $S$ is CD.
\end{prevproof}

\begin{prevproof}{Lemma}{lem:pcd}
By our assumption that $\mathcal{L}$ is NP-hard under
polynomial-time many-one reductions, there is
such a reduction from {\sc 3sat} to $\mathcal{L}$.
Since the running time of the reduction is 
bounded by a polynomial $p(n)$, we know that it
transforms a {\sc 3sat} instance of size $n$ into
an $\mathcal{L}$ instance of size at most $p(n)$.
Assume without loss of generality that $p(n)$ is
an increasing function of $n$.

Let $S$ be the set of all $n$ such that 
$\{p(n)+1, p(n)+2, \ldots, p(n+1)\}$ intersects $T$.
The set $S$ is complexity-defying, because for
any $n \in S$ we can construct a polynomial-sized
circuit that correctly decides {\sc 3sat} instances
of size $n$, as follows.  First, we take the given
{\sc 3sat} instance and apply the reduction from the
preceding paragraph to transform it into an 
$\mathcal{L}$ instance of size at most $p(n)$.
Then, letting $m$ be any element of 
$T \cap \{p(n)+1,\ldots,p(n+1)\}$, we 
apply the padding reduction to transform 
this $\mathcal{L}$ instance into another
$\mathcal{L}$ instance of size $m$.  Finally,
we solve this instance using a circuit of size
$\poly(m)$ that correctly decides $\mathcal{L}$
on all instances of size $m$; such a circuit
exists by our assumption on $T$.

For every $m \in T$ there is an $n \in \NN$ such
that $p(n) < m \leq p(n+1)$, and this $n$ belongs
to $S$.  Thus, $T \subseteq \bigcup_{n \in S} [n,p(n+1)],$
and this confirms that $T$ is PCD.
\end{prevproof}


\begin{prevproof}{Lemma}{lem:not-pcd}
Suppose that 
\begin{equation} \label{eq:not-pcd}
\NN \subseteq \bigcup_{n \in S} [n,p(n)]
\end{equation} 
for some
complexity-defying set $S$ and polynomial function $p(n)$.  We
may assume without loss of generality that $p(n)$ is an increasing
function of $n$ and that $p(n) \geq n$ for all $n$.  

Suppose that $\{\circuit_n\}$ is a polynomial-sized circuit family
that correctly decides {\sc 3sat} whenever the input size is in $S$.
We will construct a polynomial-sized circuit family that correctly
decides {\sc 3sat} on all inputs.  The construction is as follows:
given an input size $m$, using \eqref{eq:not-pcd} we may find a
natural number $n$ such that $n \leq p(m) \leq p(n).$  Since 
$p$ is an increasing function, we know that $n \geq m$.  Given
an instance of {\sc 3sat} of size $m$, we first adjoin irrelevant
clauses that don't affect its satisfiability --- e.g. the clause
$(x \vee \overline{x})$ --- until the input size is increased 
to $n$.  This
transformation can be done by a circuit of size $\poly(m)$,
since $n \leq p(m).$  Then we solve the new {\sc 3sat} instance
using the circuit $\circuit_n.$  By our assumption on $S$, this
correctly decides the original {\sc 3sat} instance of size $m$.
As $m$ was arbitrary, this establishes that $\compclass{NP}
\subseteq \compclass{P}/\poly$, as desired.
\end{prevproof}



\section{Additional Preliminaries}
\label{app:prelim}

\subsection{Truthfulness}

A $k$-bidder, $m$-item mechanism for combinatorial auctions with valuations in
$\C$ is a pair $(f,p)$ where $f:\C_m^k \rightarrow \X([m],[k])$ is an
\emph{allocation rule}, and $p=(p_1,\cdots, p_k)$ where $p_i:\C_m^k\rightarrow \mathbb
\RR$ is a \emph{payment scheme}. $(f,p)$ might be either randomized or
deterministic.

We say deterministic mechanism $(f,p)$ is \emph{truthful} if for all
$i$, all $v_i, v'_i$ and all $v_{-i}$ we have that
$v_i(f(v_i,v_{-i})_i)-p_i(v_i,v_{-i})\geq
v'_i(f(v'_i,v_{-i})_i)-p(v'_i,v_{-i})$. A randomized mechanism $(f,p)$
is \emph{universally truthful} if it is a probability distribution
over truthful deterministic mechanisms. More generally, $(f,p)$ is
\emph{truthful in expectation} if for all $i$, all $v_i, v'_i$ and all
$v_{-i}$ we have that $E[v_i(f(v_i,v_{-i})_i)-p(v_i,v_{-i})]\geq
E[v'_i(f(v'_i,v_{-i})_i)-p_i(v'_i,v_{-i})]$, where the expectation is
taken over the internal random coins of the algorithm.




\subsection{Algorithms and Approximation}

Fix a valuation class $\C$. An algorithm $\A$ for combinatorial
auctions with $\C$ valuations takes as input the number of players
$k$, the number of items $m$, and a player valuation profile
$v_1,\ldots,v_k$ where $v_i \in \C_m$. $\A$ must then output an
allocation of $[m]$ to $[k]$. For each $k$ and $m$, $\A$ induces an
allocation rule of $m$ items to $k$ bidders. Our approximation bounds
are all in terms of the number of players. Therefore, in our proofs we
consider combinatorial auctions with a fixed number of bidders $k$.

We say an algorithm $\A$ for $k$-player combinatorial auctions
\emph{achieves an $\alpha$-approximation} if, for every input $m$ and
$v_1,\ldots,v_k$:

\[ E [ v(A(m,v_1,\ldots,v_k)) ] \geq \alpha \max_{S \in \X([m],[k])} v(S)\]

Moreover, we say$\A$ \emph{achieves an $\alpha$-approximation for $m$
  items} if the above holds whenever the number of items is fixed at $m$.

\subsection{A Primer on Non-Uniform Computation}
\label{sec:nonuniform}
Non-uniform computation is a standard notion from complexity theory
(see e.g. \cite{AB09}). We say an algorithm is non-uniform if it takes
in an extra parameter, often referred to as an \emph{advice string}.
However, the advice string is allowed to vary only with the size of
the input (i.e. with $m$). Moreover, the length of the advice string
can grow only polynomially in the size of the input. If a problem
admits a non-uniform polynomial-time algorithm, this is equivalent to
the existence of a family of polynomial-sized boolean circuits for the
problem. When we say a non-uniform algorithm is polynomial-time MIWR,
we mean that the algorithm runs in time polynomial in $m$, and
maximizes over a weighted range, regardless of the advice string. When
we say a non-uniform algorithm achieves an approximation ratio of
$\alpha$ on $m$, we mean that there exists a choice of advice string
for input length $m$ such that the algorithm always outputs an
$\alpha$-approximate allocation.

\end{document}